\documentclass[11pt,letterpaper]{article}

\usepackage{amsmath,amsfonts,amsthm,amssymb,stmaryrd,relsize}
\theoremstyle{plain}

\numberwithin{equation}{section}

\newtheorem{thm}{Theorem}[section]

\theoremstyle{definition}  %% 10.20 from CMJ
\newtheorem{exam}{Example}  

\allowdisplaybreaks  % introduced 7.15.15

\newcommand{\complex}{{\mathbb C}}
\newcommand{\Natural}{{\mathbb N}}

\newcommand{\rmtr}{\mathrm{tr\,}}

\newcommand{\escript}{\mathcal{E}}
\newcommand{\lscript}{\mathcal{L}}
\newcommand{\pscript}{\mathcal{P}}
\newcommand{\sscript}{\mathcal{S}}

\newcommand{\ab}[1]{\left|#1\right|}
\newcommand{\doubleab}[1]{\left|\left|#1\right|\right|}
\newcommand{\brac}[1]{\left\{#1\right\}}
\newcommand{\paren}[1]{\left(#1\right)}
\newcommand{\sqbrac}[1]{\left[#1\right]}
\newcommand{\parensq}[1]{{\left(#1\right]}} % \left (and \right ] cause system to look for \left ( and \right ]
\newcommand{\elbows}[1]{{\left\langle#1\right\rangle}}

\newcommand{\bbar}{\overline{b}}

\begin{document}

\title{TIME EVOLUTION OF\\QUANTUM EFFECTS}
\author{Stan Gudder\\ Department of Mathematics\\
University of Denver\\ Denver, Colorado 80208\\
sgudder@du.edu}
\date{}
\maketitle

\begin{abstract}
For quantum effects $a$ and $b$ we define the $a$-evolution of $b$ at time $t$ denoted by $b(t\mid a)$. We interpret $b(t\mid a)$ as the influence that $a$ has on $b$ at time $t$ when $a$ occurs, but is not measured at time $t=0$. Using $b(t\mid a)$ we define the time-dependent sequential product $a[t]b$. This is interpreted as an effect that results from first measuring $a$ and then measuring $b$ after a time delay $t$. Various properties of $a[t]b$ are derived and it is shown that $a[t]b$ is constant in time if and only if $a$ and $b$ commute or $a$ is a multiple of a projection. These concepts are extended to observables for a quantum system. The ideas are illustrated with some examples.
\end{abstract}

\section{Basic Definitions}  % Section 1
We consider quantum systems represented by a finite-dimensional complex Hilbert space $H$ and denote the set of linear operators on $H$ by
$\lscript (H)$. For $A,B\in\lscript (H)$ we write $A\le B$ if $\elbows{\psi ,A\psi}\le\elbows{\psi ,B\psi}$ for all $\psi\in H$. We call $a\in\lscript (H)$ an \textit{effect} if $0\le a\le I$ where $0,I$ are the zero and identity operators, respectively \cite{blm96,hz12,nc00}. Effects represent two outcome measurements and are frequently called yes-no experiments. We denote the set of effects on $H$ by $\escript (H)$. If $a\in\escript (H)$ is measured and the result is yes, we say that $a$ \textit{occurs}. An element $\rho\in\escript (H)$ with trace $\rmtr(\rho)=1$ is called a
\textit{state} and the set of states on $H$ is denoted by $\sscript (H)$. The \textit{probability that} $a$ \textit{occurs} when the system is in the state $\rho$ is defined by $\pscript _\rho (a)=\rmtr (\rho a)$. For $a,b\in\escript (H)$ we define the \textit{sequential product} $a\circ b=a^{1/2}ba^{1/2}$ where $a^{1/2}$ is the unique positive square-root of $a$. It is easy to check that $a\circ b\in\escript (H)$ and it can be shown that
$a\circ b=b\circ a$ if and only if $a$ and $b$ commute $(ab=ba)$ \cite{gg02,gl08}. We interpret $a\circ b$ as the effect that results from first measuring $a$ and then measuring $b$ directly afterwards. Since $a$ is measured first, its measurement can interfere with the measurement of $b$ but not vice-versa \cite{gg02,gl08,jj09,wj18,wet18}.

An \textit{observable} for a system described by $H$ is a finite set of effects $A=\brac{A_x\colon x\in\Omega _A}\subseteq\escript (H)$ satisfying $\sum _{x\in\Omega _A}A_x=I$. We call $\Omega _A$ the \textit{outcome set} for $A$ and the effect $A_x$, $x\in\Omega _A$, occurs if $A$ has outcome $x$ when measured. The \textit{distribution} for $A$ in the state $\rho\in\sscript (H)$ is the probability measure
$\Phi _\rho ^A(x)=\rmtr (\rho A_x)$. Notice that $\Phi _\rho ^A$ gives a probability measure because
\begin{equation*}
\sum _{x\in\Omega _A}\Phi _\rho ^A(x)=\rmtr\paren{\rho\sum _{x\in A}A_x}=\rmtr (\rho )=1
\end{equation*}
If $A=\brac{A_x\colon x\in\Omega _A}$ and $B=\brac{B_y\colon y\in\Omega _B}$ are observables we define the \textit{sequential product}
$A\circ B=\brac{A_x\circ B_y\colon (x,y)\in\Omega _A\times\Omega _B}$. Thus, the outcome set for $A\circ B$ is\
$\Omega _{A\circ B}=\Omega _A\times\Omega _B$ and corresponding to outcome $(x,y)$ we have the effect $A\circ B_{(x,y)}=A_x\circ B_y$. Notice that $A\circ B$ is indeed an observable because
\begin{align*}
\sum\brac{A_x\circ B_y\colon (x,y)\in\Omega _A\times\Omega _B}&=\sum _{x\in\Omega _A}A_x\circ\sum _{y\in\Omega _B}B_y
  =\sum _{x\in\Omega _A}A_x\circ I\\
  &=\sum _{x\in\Omega _A}A_x=I
\end{align*}
We also define the observable $B$ \textit{conditioned by} $A$ to be the set of effects $(B\mid A)_y=\sum\limits _{x\in\Omega _A}A_x\circ B_y$ with outcome set $\Omega _B$. Again, $(B\mid A)$ is an observable because
\begin{equation*}
\sum _{y\in\Omega _B}(B\mid A)_y=\sum _{x\in\Omega _A}A_x\circ\sum _{y\in\Omega _B}B_y=I
\end{equation*}

Two effects $a,b\in\escript (H)$ \textit{coexist} \cite{blm96,hz12} if there exist effects $a_1,b_1,c\in\escript (H)$ such that $a_1+b_1+c\le I$ and
$a=a_1+c$, $b=b_1+c$. If $a,b$ coexist we can define $d\in\escript (H)$ by $d=I-a_1-b_1-c$. Then $A=\brac{a_1,b_1,c,d}$ is an observable and we can view $a$ and $b$ as corresponding to outcomes of $A$. In this way, we can simultaneously measure $a$ and $b$ using the single observable $A$.

\section{Time Evolutions}  % Section 2
For $a\in\escript (H)$, $\rho\in\sscript (H)$ if $a$ occurs and a time $t$ elapses, we interpret $e^{ita}\rho e^{-ita}$ as the resulting state at time
$t$. A way of viewing this is that $\rho\mapsto e^{ita}\rho e^{-ita}$ is a unitary $a$-\textit{channel} \cite{hz12,nc00} at time $t$ and
$e^{ita}\rho e^{-ita}$ results from sending $\rho$ through this channel. After $\rho$ is sent through this $a$-channel, the probability that an effect
$b\in\escript (H)$ occurs in the resulting state becomes:
\begin{equation*}
\rmtr (e^{ita}\rho e^{-ita}b)=\rmtr (\rho e^{-ita}be^{ita})
\end{equation*}
We call $b(t\mid a)=e^{-ita}be^{ita}$, $t\in(-\infty ,\infty )$ the $a$-\textit{evolution} of $b$. We interpret $b(t\mid a)$ as the influence that $a$ has on $b$ at time $t$ when $a$ occurs at time $t$. Notice that $b(t\mid a)\in\escript (H)$ for all $t\in (-\infty ,\infty )$. The \textit{rate of change} of
$b(t\mid a)$ becomes:
\begin{align}             % equation (2.1)
\label{eq21}
\tfrac{d}{dt}\,b(t\mid a)&=-iae^{-ita}be^{ita}+e^{-ita}b(ia)e^{ita}\notag\\
   &ie^{-ita}\sqbrac{b,a}e^{ita}=i\sqbrac{b(t\mid a),a}
\end{align}
where $\sqbrac{b,a}=ba-ab$ is the \textit{commutant} of $b$ with $a$. It follows that $b(t\mid a)=b$ for all $t\in (-\infty ,\infty )$ if and only if
$\sqbrac{b,a}=0$. Although $i\sqbrac{b(t\mid a),a}$ is self-adjoint, it need not be positive. That is, $i\sqbrac{b(t\mid a),a}\not\ge 0$ in general. This is illustrated in the following example.

\begin{exam}{1}  % Example 1
Let $a=\begin{bmatrix}1&0\\0&1/2\end{bmatrix},b=\tfrac{1}{2}\begin{bmatrix}1&1\\1&1\end{bmatrix}\in\escript (\complex ^2)$ be qubit effects. We have that
\begin{align*}
b(t\mid a)&=e^{-ita}be^{ita}=\frac{1}{2}\,\begin{bmatrix}e^{-it}&0\\0&e^{-it/2}\end{bmatrix}\begin{bmatrix}1&1\\1&1\end{bmatrix}
    \begin{bmatrix}e^{it}&0\\0&e^{it/2}\end{bmatrix}\\\noalign{\smallskip}
    &=\frac{1}{2}\,\begin{bmatrix}1&e^{-it/2}\\e^{it/2}&1\end{bmatrix}
\end{align*}
and
\begin{align*}
i\sqbrac{b(t\mid a),a}&=\frac{i}{2}\,
  \brac{\begin{bmatrix}1&e^{-it/2}\\e^{it/2}&1\end{bmatrix}\begin{bmatrix}1&0\\0&1/2\end{bmatrix}-\begin{bmatrix}1&0\\0&1/2\end{bmatrix}%
  \begin{bmatrix}1&e^{-it/2}\\e^{it/2}&1\end{bmatrix}}\\\noalign{\smallskip}
  &=\frac{i}{2}\,
  \brac{\begin{bmatrix}1&\tfrac{1}{2}\,e^{-it/2}\\e^{it/2}&1/2\end{bmatrix}-\begin{bmatrix}1&e^{-it/2}\\\tfrac{1}{2}\,e^{it/2}&1/2\end{bmatrix}}\\
  &=\frac{i}{4}\,\begin{bmatrix}0&-e^{-it/2}\\e^{it/2}&0\end{bmatrix}
\end{align*}
Then the eigenvalues of $i\sqbrac{b(t\mid a),a}$ are $\pm 1/4$ so $i\sqbrac{b(t\mid a),a}$ is not positive.\hfill\qedsymbol
\end{exam}

It follows from \eqref{eq21} that 
\begin{equation*}
\frac{d^2}{dt^2}\,b(t\mid a)=\frac{d}{dt}\,\frac{db}{dt}\,(t\mid a)=-\sqbrac{\sqbrac{b(t\mid a),a},a}
\end{equation*}
and continuing we obtain
\begin{equation*}
\frac{d^n}{dt^n}b(t\mid a)=i^{(n)}\sqbrac{\sqbrac{\sqbrac{b(t\mid a),a},a}\cdots ,a}
\end{equation*}
\smallskip

\begin{exam}  % Example 2
The simplest nontrivial example if an $a$-evolution is when $a=\lambda p$ where $p\ne 0$, I is a projection and $\lambda\in\parensq{0,1}$. We then have that
\begin{align}             % equation (2.2)
\label{eq22}
b(t\mid a)&=e^{-i\lambda tp}be^{i\lambda tp}=(e^{-i\lambda t}p+I-p)b(e^{i\lambda t}p+I-p)\notag\\
   &=pbp+(I-p)be^{i\lambda t}p+e^{-i\lambda t}pb(I-p)+(I-p)b(I-p)\notag\\
   &=b+2(1-\cos\lambda t)pbp+(e^{-i\lambda t}-1)pb+(e^{i\lambda t}-1)bp
\end{align}
This clearly shows the deviation of $b(t\mid a)$ from $b$. In terms of the norm, a measure of this deviation becomes $\doubleab{b(t\mid a)-b}$. For example, suppose $H=\complex ^2$ is the qubit Hilbert space and
\begin{equation*}
p=\begin{bmatrix}1&0\\0&0\end{bmatrix},\quad b=\begin{bmatrix}b_{11}&b_{12}\\{\bbar}_{12}&b_{22}\end{bmatrix}
\end{equation*}
We then have that
\begin{equation*}
pbp=\begin{bmatrix}b_{11}&0\\0&0\end{bmatrix},\quad pb=\begin{bmatrix}b_{11}&b_{12}\\0&0\end{bmatrix},
   \quad bp=\begin{bmatrix}b_{11}&0\\\bbar _{12}&0\end{bmatrix}
\end{equation*}
and by \eqref{eq22} we obtain
\begin{equation*}
b(t\mid a)-b=\begin{bmatrix}0&(e^{-i\lambda t}-1)b_{12}\\(e^{i\lambda t}-1)\bbar _{12}&0\end{bmatrix}
\end{equation*}
The eigenvalues of $b(t\mid a)-a$ become $\pm\sqrt{2(1-\cos\lambda t)\,}\ab{b_{12}}$ so that
\begin{equation*}
\doubleab{b(t\mid a)-b}=\sqrt{2(1-cos\lambda t)\,}\ab{b_{12}}
\end{equation*}
The maximum deviation is obtained when $t=\pm\frac{(2n-1)\pi}{\lambda}$, $n\in\Natural$, in which case
$\doubleab{b(t\mid a)-b}=2\ab{b_{12}}$. We also have that
\begin{equation*}
b(t\mid a)=\begin{bmatrix}b_{11}&e^{-i\lambda t}b_{12}\\e^{i\lambda t}\bbar _{12}&b_{22}\end{bmatrix}
\end{equation*}
so that
\begin{align*}
\frac{d}{dt}\,b(t\mid a)&=i\sqbrac{b(t\mid a),a}\\
  &=i\lambda\brac{\begin{bmatrix}b_{11}&e^{-i\lambda t}b_{12}\\e^{i\lambda t}\bbar _{12}&b_{22}\end{bmatrix}%
  \begin{bmatrix}1&0\\0&0\end{bmatrix}-\begin{bmatrix}1&0\\0&0\end{bmatrix}%
  \begin{bmatrix}b_{11}&e^{-i\lambda t}b_{12}\\e^{i\lambda t}\bbar _{12}&b_{22}\end{bmatrix}}\\\noalign{\smallskip}
  &=i\lambda\brac{\begin{bmatrix}b_{11}&0\\e^{i\lambda t}\bbar _{12}&0\end{bmatrix}%
  \!-\!\begin{bmatrix}b_{11}&e^{-i\lambda t}b_{12}\\0&0\end{bmatrix}}
  =i\lambda\begin{bmatrix}0&-e^{-i\lambda t}b_{12}\\e^{i\lambda t}\bbar _{12}&0\end{bmatrix}
\end{align*}
The eigenvalues of $\tfrac{d}{dt}b(t\mid a)$ are $\pm\lambda\ab{b_{12}}$ which is independent of time.\hfill\qedsymbol
\end{exam}

\section{Sequential Products}  % Section 3
When we defined the sequential product $a\circ b=a^{1/2}ba^{1/2}$ in Section~1, we assumed that $a$ was measured first but that there was essentially no time delay between the measurements of $a$ and $b$. To incorporate such a time delay, we define the \textit{time dependent sequential product} $a\sqbrac{t}b$ by employing the $a$-evolution of $a\circ b$ so that
\begin{equation*}
a\sqbrac{t}b=(a\circ b)(t\mid a)=e^{-ita}a\circ be^{ita}=a\circ e^{-ita}be^{ita}=a\circ\sqbrac{b(t\mid a)}
\end{equation*}
We interpret $a\sqbrac{t}b$ as the effect that results from first measuring $a$ and then measuring $b$ after a time delay $t$. As in \eqref{eq21} we have the \textit{rate of change} of $a\sqbrac{t}b$ given by the following
\begin{align}                % equation (3.1)
\label{eq31}
\tfrac{d}{dt}\,a\sqbrac{t}b&=-iae^{-ita}a^{1/2}ba^{1/2}e^{ita}+e^{-ita}a^{1/2}ba^{1/2}(ia)e^{ita}\notag\\
   &=-ie^{-ita}a^{3/2}ba^{1/2}e^{ita}+ie^{-ita}a^{1/2}ba^{3/2}e^{ita}\notag\\
   &=i\sqbrac{a\sqbrac{t}b,a}
\end{align}
We say that $a\sqbrac{t}b$ is \textit{constant} if $a\sqbrac{t}b=a\circ b$ for all $t\in (-\infty ,\infty )$. Of course $a\sqbrac{t}b$ is constant if and only if $\tfrac{d}{dt}\,a\sqbrac{t}b=0$ which by \eqref{eq31} is equivalent to $\sqbrac{a\sqbrac{t}b,a}=0$ for all $t\in (-\infty ,\infty )$. We have seen that $b(t\mid a)$ is constant if and only if $\sqbrac{a,b}=0$. It is surprising that the result for $a\sqbrac{t}b$ is slightly different.

\begin{thm}    % Theorem 3.1
\label{thm31}
The following statements are equivalent:\newline
{\rm{(i)}}\enspace $a\sqbrac{t}b$ is constant,
{\rm{(ii)}}\enspace $\sqbrac{a\circ b,a}=0$,
{\rm{(iii)}}\enspace either $\sqbrac{a,b}=0$ or $a=\lambda p$, where $p$ is a projection and $\lambda\in\sqbrac{0,1}$.
\end{thm}
\begin{proof}
It follows from \eqref{eq31} that $a\sqbrac{t}b$ is constant if and only if $\sqbrac{a\sqbrac{t}b,a}=0$ for all $t\in (-\infty ,\infty )$. As in \eqref{eq31} this is equivalent to
\begin{equation*}
e^{-ita}a^{3/2}ba^{1/2}e^{ita}=e^{-ita}a^{1/2}ba^{3/2}e^{ita}
\end{equation*}
for all $t\in (-\infty ,\infty )$. Multiplying by $e^{ita}$ on the left and $e^{-ita}$ on the right gives $a^{3/2}ba^{1/2}=a^{1/2}ba^{3/2}$ which is equivalent to (ii). Hence (i) and (ii) are equivalent. We now show that (i) and (iii) are equivalent. If $\sqbrac{a,b}=0$, then clearly (i) holds and if
$a=\lambda p$, then (ii) holds so again (i) holds. Finally, suppose that (i) holds. By the spectral theorem $a=\sum\limits _{i=1}^n\lambda _ip_i$ where $\lambda _i\in\parensq{0,1}$, $\lambda _i\ne\lambda _j$ when $i\ne j$ and $p_i$ are projections. We then have that
\begin{equation*}
a\circ b=\sum _{i=1}^n\lambda _i^{1/2}p_ib\sum _{j=1}^n\lambda _j^{1/2}p_j=\sum _{i,j=1}^n\lambda _i^{1/2}\lambda _j^{1/2}p_ibp_j
\end{equation*}
Hence,
\begin{align*}
a(a\circ b)&=\sum _{k=1}^n\lambda _kp_k\sum _{i,j=1}^n\lambda _i^{1/2}\lambda _j^{1/2}p_ibp_j
  =\sum _{k,j=1}^n\lambda _k^{3/2}\lambda _j^{1/2}p_kbp_j\\
\intertext{and}
(a\circ b)a&=\sum _{i,j=1}^n\lambda _i^{1/2}\lambda _j^{1/2}p_ibp_j\sum _{k=1}^n\lambda _kp_k
  =\sum _{k,i=1}^n\lambda _i^{1/2}\lambda _k^{3/2}p_ibp_k
\end{align*}
Since (ii) holds, multiplying on the right by $p_r$ gives
\begin{equation*}
\sum _{k=1}^n\lambda _k^{3/2}\lambda _r^{1/2}p_kbp_r=\sum _{i=1}^n\lambda _i^{1/2}\lambda _k^{3/2}p_ibp_r
\end{equation*}
Multiplying on the left by $p_s$ we obtain
\begin{equation}             % equation (3.2)
\label{eq32}
\lambda _s^{3/2}\lambda _r^{1/2}p_sbp_r=\lambda _s^{1/2}\lambda _r^{3/2}p_sbp_r
\end{equation}
If $n=1$, then $a=\lambda _1p_1$ so (iii) holds and we are finished. Now suppose that $n\ne 1$. If $s\ne r$ and $p_sbp_r\ne 0$, then by \eqref{eq32} we have that $\lambda _s^{3/2}\lambda _r^{1/2}=\lambda _s^{1/2}\lambda _r^{3/2}$ which gives $\lambda _s=\lambda _r$. But then $s=r$, which is a contradiction. Hence, $p_sbp_r=0$ whenever $s\ne r$. We conclude that
\begin{equation*}
0=\sum _{r\ne s}p_sbp_r=p_sb(I-p_s)=p_sb-p_sbp_s
\end{equation*}
But then $p_sb=p_sbp_s$ and taking adjoints gives $p_sb=bp_s$, $s=1,2,\ldots ,n$. It follows that $\sqbrac{a,b}=0$ so (iii) holds.
\end{proof}

\begin{exam}  % Example 3
We now show directly that if $a=\lambda p$ with $\lambda\in\sqbrac{0,1}$ and $p$ a projection, then $a\sqbrac{t}b$ is constant. As in Example~2, we have that 
\begin{align*}
a\sqbrac{t}b&=e^{-i\lambda tp}(\lambda p)\circ be^{i\lambda tp}=\lambda e^{-i\lambda tp}pbpe^{i\lambda tp}\\
   &=\lambda pbp+2(1-\cos\lambda t)\lambda pbp+(e^{-i\lambda t}-1)\lambda pbp+(e^{i\lambda t}-1)\lambda pbp\\
   &=\lambda pbp=a\circ b\hskip 21pc\square
\end{align*}
\end{exam}

It is clear that if $\sqbrac{a,b}=0$, then $a\sqbrac{t}b=b\sqbrac{t}a$ for all $t\in (-\infty ,\infty )$. Conversely, if $a\sqbrac{t}b=b\sqbrac{t}a$ for all $t\in (-\infty ,\infty )$ then letting $t=0$ gives $a\circ b=b\circ a$. It then follows that $\sqbrac{a,b}=0$ \cite{gg02,gl08}. However, we do not know whether $a\sqbrac{t}b=b\sqbrac{t}a$ for some $t\in (-\infty ,\infty )$ implies that $\sqbrac{a,b}=0$.

The next theorem summarizes some properties of the time evolution.

\begin{thm}    % Theorem 3.2
\label{thm32}
For $a,b,c\in\escript (H)$ the following statements hold.\newline
{\rm{(i)}}\enspace $b(t_1+t_2\mid a)=\sqbrac{b(t_1\mid a)}(t_2\mid a)$.
{\rm{(ii)}}\enspace $a\sqbrac{t_1+t_2}b=\paren{a\sqbrac{t_1}b}(t_2\mid a)$.
{\rm{(iii)}}\enspace $a\sqbrac{t}b\le a$.
{\rm{(iv)}}\enspace If $a+b\le I$, then $(a+b)(t\mid c)=a(t\mid c)+b(t\mid c)$ and $c\sqbrac{t}(a+b)=c\sqbrac{t}a+c\sqbrac{t}b$.
{\rm{(v)}}\enspace $a(t\mid c)\circ b(t\mid c)=a\circ b(t\mid c)$.
{\rm{(vi)}}\enspace $a\sqbrac{t}\paren{b\sqbrac{t}c}=a\circ (e^{-ita}e^{-itb}b\circ ce^{itb}e^{ita})$ and if $\sqbrac{a,b}=0$,
then $a\sqbrac{t}\paren{b\sqbrac{t}c}=e^{-it(a+b)}(a\circ b)\circ ce^{it(a+b)}$.
{\rm{(vii)}}\enspace If $a$ and $b$ coexist, then $a(t\mid c)$, $b(t\mid c)$ coexist and $c\sqbrac{t}a$, $c\sqbrac{t}b$ coexist.
\end{thm}
\begin{proof}
(i)\enspace For all $t_1,t_2\in (-\infty ,\infty )$ we have that
\begin{align*}
b(t_1+t_2\mid a)&=e^{-i(t_1+t_2)a}be^{i(t_1+t_2)a}=e^{-it_2a}e^{-it_1a}be^{it_1a}e^{it_2a}\\
   &=e^{-it_2a}b(t_1\mid a)e^{it_2a}=\sqbrac{b(t_1\mid a)}(t_2\mid a)
\end{align*}
(ii)\enspace For all $t_1,t_2\in (-\infty ,\infty )$ we have that
\begin{align*}
a\sqbrac{t_1+t_2}b&=e^{-i(t_1+t_2)a}a\circ be^{i(t_1+t_2)a}=e^{-it_2a}e^{-it_1a}a\circ be^{it_1a}e^{it_2a}\\
   &=e^{-it_2a}a\sqbrac{t_1}be^{it_2a}=\paren{a\sqbrac{t_1}b}(t_2\mid a)
\end{align*}
(iii)\enspace Since $a\circ b\le a$, we have that
\begin{equation*}
a\sqbrac{t}b=e^{-ita}a\circ be^{ita}\le e^{-ita}ae^{ita}=a
\end{equation*}
(iv)\enspace If $a+b\le I$, then $a+b\in\escript (H)$ and we obtain
\begin{align*}
(a+b)(t\mid c)&=e^{-itc}(a+b)e^{itc}=e^{-itc}ae^{itc}+e^{-itc}be^{itc}=a(t\mid c)+b(t\mid c)\\
  c\sqbrac{t}(a+b)&=e^{-itc}c\circ (a+b)e^{itc}=e^{-itc}c\circ ae^{itc}+e^{-itc}c\circ be^{itc}\\
  &=c\sqbrac{t}a+c\sqbrac{t}b
\end{align*}
(v)\enspace For all $t\in (-\infty ,\infty )$ we have that
\begin{align*}
a(t\mid c)\circ b(t\mid c)&=(e^{-itc}ae^{itc})\circ (e^{-itc}be^{itc})\\
   &=e^{-itc}a^{1/2}e^{itc}e^{-itc}be^{itc}e^{-itc}a^{1/2}e^{itc}\\
   &=e^{-itc}a\circ be^{itc}=(a\circ b)(t\mid c)
\end{align*}
(vi)\enspace For all $t\in (-\infty ,\infty )$ we obtain
\begin{align*}
a\sqbrac{t}\paren{b\sqbrac{t}c}&=e^{-ita}a\circ\paren{b\sqbrac{t}c}e^{ita}=e^{-ita}a\circ (e^{-itb}b\circ ce^{itb})e^{ita}\\
   &=a\circ (e^{-ita}e^{-itb}b\circ ce^{itb}e^{ita})
\end{align*}
and if $\sqbrac{a,b}=0$, then
\begin{equation*}
a\sqbrac{t}\paren{b\sqbrac{t}c}=e^{-it(a+b)}a\circ (b\circ c)e^{it(a+b)}=e^{-it(a+b)}(a\circ b)\circ ce^{it(a+b)}
\end{equation*}
(vii)\enspace Since $a$ and $b$ coexist, there exist effects $a_1,b_1,d\in\escript (H)$ such that $a_1+b_1+d\le I$ and $a=a_1+d$, $b=b_1+d$. But then
\begin{equation*}
a_1(t\mid c)+b_1(t\mid c)+d(t\mid c)=e^{-itc}a_1e^{itc}+e^{-itc}b_1e^{itc}+e^{-itc}de^{itc}\le I
\end{equation*}
and by (iv)
\begin{align*}
a(t\mid c)&=a_1(t\mid c)+d(t\mid c)\\
b(t\mid c)&=b_1(t\mid c)+d(t\mid c)
\end{align*}
Hence, $a(t\mid c)$ and $b(t\mid c)$ coexist. Moreover, we have that
\begin{equation*}
c\sqbrac{t}a_1+c\sqbrac{t}b_1+c\sqbrac{t}d=e^{-itc}c\circ (a_1+b_1+d)e^{itc}\le I
\end{equation*}
and by (iv)
\begin{align*}
c\sqbrac{t}a&=c\sqbrac{t}a_1+c\sqbrac{t}d\\
c\sqbrac{t}b&=c\sqbrac{t}b_1+c\sqbrac{t}d
\end{align*}
Hence, $c\sqbrac{t}a$ and $c\sqbrac{t}b$ coexist.
\end{proof}

We can extend the concepts of time evolution and time dependent sequential products of effects to observables. If $a\in\escript (H)$ and
$B=\brac{B_y\colon y\in\Omega _B}$ is an observable on $H$ the $a$-\textit{evolution} of $B$ is the observable $B(t\mid a)$ given by
$\Omega _{B(t\mid a)}=\Omega _B$ for all $t\in (-\infty ,\infty )$ and
\begin{equation*}
B(t\mid a)_y=e^{-ita}B_ye^{ita}
\end{equation*}
for all $y\in\Omega _B$. For all $\rho\in\sscript (H)$, the distribution of $B(t\mid a)$ becomes
\begin{equation*}
\Phi _\rho ^{B(t\mid a)}(y)=\rmtr\sqbrac{\rho B(t\mid a)_y}=\rmtr (\rho e^{-ita}B_ye^{ita})=\rmtr (e^{ita}\rho e^{-ita}B_y)
\end{equation*}
For observables $A=\brac{A_x\colon x\in\Omega _A}$, $B=\brac{B_y\colon y\in\Omega _B}$, the \textit{time dependent sequential product} is the observable $A\sqbrac{t}B$ given by $\Omega _{A\sqbrac{t}B}=\Omega _A\times\Omega _B$ and 
\begin{align*}
A\sqbrac{t}B_{x,y)}&=e^{-itA_x}A\circ B_{(x,y)}e^{itA_x}=A_x^{1/2}e^{-itA_x}B_ye^{itA_x}A_x^{1/2}\\
   &=A_x\circ B(t\mid A_x)_y
\end{align*}
Notice that $A\sqbrac{t}B$ is indeed an observable for all $t\in (-\infty ,\infty )$ because
\begin{align*}
\sum _{x,y}A\sqbrac{t}B_{(x,y)}&=\sum _{x,y}A_x^{1/2}e^{-itA_X}B_ye^{itA_x}A_x^{1/2}\\
   &=\sum _xA_x^{1/2}e^{-itA_x}\sum _yB_ye^{itA_x}A_x^{1/2}\\
   &=\sum _xA_x^{1/2}e^{-itA_x}Ie^{itA_x}A_x^{1/2}=\sum _xA_x=I
\end{align*}
The distribution of $A\sqbrac{t}B$ in the state $\rho$ becomes
\begin{equation*}
\Phi _\rho ^{A\sqbrac{t}B}(x,y)=\rmtr (e^{itA_x}\rho e^{-itA_x}A_x\circ B_y)=\rmtr\sqbrac{e^{itA_x}\rho e^{-itA_x}A\circ B_{(x,y)}}
\end{equation*}
Finally, the \textit{time dependent conditional observable} $(B\mid A)(t\mid A)$ is given by $\Omega _{(B\mid A)(t\mid A)}=\Omega _B$ and
\begin{equation*}
(B\mid A)(t\mid A)_y=\sum _xA_x\sqbrac{t}B_y=\sum _xA_x\circ B(t\mid A_x)_y
\end{equation*}
The distribution of $(B\mid A)(t\mid A)$ in the state $\rho$ becomes
\begin{equation*}
\Phi _\rho ^{(B\mid A)(t\mid A)}(y)=\sum _x(e^{itA_x}\rho e^{-itA_x}A_x\circ B_y)
\end{equation*}

Straightforward generalizations of Theorems~\ref{thm31} and \ref{thm32} hold for observables. For example, consider the second part of
Theorem~\ref{thm32}(iv). If $B_i$, $i=1,2,\ldots ,n$, are observables on $H$ with outcome sets $\Omega _{B_i}=\Omega$, $i=1,2,\ldots ,n$, and $\lambda _i\in\sqbrac{0,1}$ satisfy $\sum\limits _{i=1}^n\lambda _i=1$, then the \textit{convex combination}
$\sum\limits _{i=1}^n\lambda _iB_i$ is the observable with outcome space $\Omega$ given by
$\paren{\sum\limits _{i=1}^n\lambda _iB_i}_y=\sum\limits _{i=1}^n\lambda _i(B_i)_y$. If $A$ is another observable on $H$, then we obtain
\begin{align*}
A\sqbrac{t}\paren{\sum\lambda _iB_i}_{(x,y)}&=e^{-itA_x}\paren{A\circ\sum\lambda _iB_i}_{(x,y)}e^{itA_x}\\
  &=e^{-itA_x}A_x\circ\sum\lambda _iB_{iy}e^{itA_x}=e^{-itA_x}\sum\lambda _iA_x\circ B_{iy}e^{itA_x}\\
  &=\sum\lambda _ie^{-itA_x}A_x\circ B_{iy}e^{-iA_x}=\sum\lambda _i(A\sqbrac{t}B_i)_{(x,y)}
\end{align*}
It follows that $A\sqbrac{t}\paren{\sum\lambda _iB_i}=\sum\lambda _iA\sqbrac{t}B_i$.

\end{document}